\documentclass[draftclsnofoot,peerreview]{IEEEtran}

\usepackage{amsmath,amsthm,amssymb}
\usepackage{graphicx}
\usepackage{esint}
\usepackage{cite}
\usepackage{subcaption}
\usepackage{color}
\usepackage{import}

\newcommand{\mb}[1]{{\boldsymbol{#1}}}
\newcommand{\mbb}[1]{\mathbb{#1}}
\newcommand{\mc}[1]{\mathcal{#1}}

\newcommand{\mr}[1]{\mathrm{#1}}
\DeclareMathOperator*{\argmin}{argmin}

\makeatletter
\theoremstyle{plain}
\newtheorem{theorem}{\protect\theoremname}
\theoremstyle{plain}
\newtheorem{corollary}[theorem]{\protect\corollaryname}
\newtheorem{remark}{\em Remark}

\makeatother

\providecommand{\corollaryname}{Corollary}
\providecommand{\theoremname}{Theorem}

\begin{document}
\title{Compressive Deconvolution\\ in Random Mask Imaging}

\author{Sohail~Bahmani and Justin~Romberg,~\IEEEmembership{Senior Member}%
\thanks{The authors are with the School of Electrical and Computer Engineering, Georgia Institute of Technology in Atlanta, GA.  E-mail: \{sohail.bahmani,jrom\}@ece.gatech.edu.  This work was supported by ONR grant N00014-11-1-0459, and NSF grants CCF-1415498 and CCF-1422540.}
}

\maketitle
%\IEEEpeerreviewmaketitle

\begin{abstract}   
We investigate the problem of reconstructing signals from a subsampled convolution of their modulated versions and a known filter. The problem is studied as applies to a specific imaging architecture that relies on spatial phase modulation by randomly coded ``masks''. The diversity induced by the random masks is deemed to improve the conditioning of the deconvolution problem while maintaining sampling efficiency.

We analyze a linear model of the imaging system, where the joint effect of the spatial modulation, blurring, and spatial subsampling is represented concisely by a measurement matrix. We provide a bound on the conditioning of this measurement matrix in terms of the number of masks $K$, the dimension (i.e., the pixel count) of the scene image $L$, and certain characteristics of the blurring kernel and subsampling operator.  The derived bound shows that stable deconvolution is possible with high probability even if the number of masks (i.e., $K$ ) is as small as $\frac{L\log L}{N}$, meaning that the total number of (scalar) measurements is within a logarithmic factor of the image size. Furthermore, beyond a critical number of masks determined by the extent of blurring and subsampling, use of every additional mask improves the conditioning of the measurement matrix.

 We also consider a more interesting scenario where the target image is known to be sparse.  We show that under mild conditions on the blurring kernel, with high probability the measurement matrix is a restricted isometry when the number of masks is within a logarithmic factor of the sparsity of the scene image. Therefore, the scene image can be reconstructed using any of the well-known sparse recovery algorithms such as the basis pursuit. The bound on the required number of masks grows linearly in sparsity of the scene image but logarithmically in its ambient dimension. The bound provides a quantitative view of the effect of the blurring and subsampling on the required number of masks, which is critical for designing efficient imaging systems.
\end{abstract}

\section{Introduction}\label{sec:Intro}

In this paper, we investigate the mathematics of reconstructing a high-resolution image from measurements made by a low-resolution sensor array.  A schematic of this type of imaging system is shown in Figure~\ref{fig:SchemaMI}: an image is focused onto a spatial light modulator (SLM), passes through a blurring lens, and the resulting intensity image is integrated over a relatively large region in space.  The blurring lens spreads the energy in the image out spatially, allowing the sensor array to have gaps.  We show that under mild conditions, the resolution of this type of system is fundamentally limited by the resolution of the spatial light modulator; the spatial resolution of sensors and the blurring lens play only a minor role.  We also show that the diversity provided by the spatial light modulator makes the deconvolution process stable.  

Our mathematical analysis uses an idealized model for this imaging system.  We model the imaging process as a linear operator that maps an image into a set of indirect measurements.  Multiple batches of these measurements are taken, each with a different pattern on the SLM.  The favorable conditioning in the inverse problem comes from using a diverse set of patterns; we show that choosing the patterns at random results in improved acquisition efficiency in both the least squares and sparse reconstruction scenarios.  We show that the total number of measurements sufficient for least squares reconstruction is within a logarithmic factor of the number of pixels in the image which is nearly-optimal. For sparse reconstruction, a total number of masks within a logarithmic factor of the sparsity suffices to achieve a stable deconvolution.  Mathematically, these are statements about the singular values of both the imaging matrix as a whole and the submatrices formed from subsets of its columns.

Our main results, Theorems 1 and 2 in Section~\ref{sec:MainResults}, give stable reconstruction guarantees in terms of the number of masks $K$ used, the number of sensors $N$ in the array, the number of pixels $L$ in the high-resolution reconstruction, and parameters that characterize the joint spatial response of the blurring and sampling system.  These mathematical results give credence to the idea that diverse spatial light modulation makes spatial deconvolution a well-posed problem, even in the presence of heavy subsampling. 

\subsection{Contributions} 
\paragraph*{Stability of deconvolution via least squares} The relation between the scene image and the measured samples in the illustrated system can be described mathematically by a matrix.  Therefore, the conditioning of this measurement matrix determines the stability of an ordinary deconvolution based on the standard least squares. We quantify the number of random masks that is sufficient to bound the conditioning of the measurement matrix. The obtained bounds show that the blurring and the subsampling determine a critical number of masks needed to guarantee that the measurement matrix is well-conditioned, and thus 
the deconvolution is stable. With every additional mask beyond this critical number the guaranteed conditioning of the measurement matrix improves.

\paragraph*{Sparse deconvolution} We also study the case where the scene image is sparse. In this scenario, as the number of samples can be much less than the number of the scene pixels, we refer to the reconstruction of the scene image as  \emph{compressive deconvolution} akin to Compressive Sensing (CS) \cite{candes_robust_2006,candes_near_2006,donoho_compressed_2006}. We characterize the sufficient number of masks to guarantee the Restricted Isometry Property (RIP) for the (scaled) measurement matrix and thus successful reconstruction of the scene image using $\ell_1$-minimization. The established bound show that depending on the effect of the blurring and the subsampling, if  the number of masks grow linearly with the sparsity of the scene image but merely logarithmically with its pixel count we can have suitable RIP constants for successful $\ell_1$-minimization. Therefore, the number of measurements can potentially be significantly smaller that the pixel count of the scene image.

\subsection{Background and related Work}
Classical deconvolution techniques can be broadly categorized in two frameworks based on their approaches to regularization of the inverse problem. Methods of the first category, including Wiener filtering and a variety of Bayesian methods, assume some stochastic model for the image or the blurring kernel that is often application specific. Methods of the second category, that are essentially some variants of the least squares, only use the deterministic spatial or spectral structures of the image such as smoothness for regularization. For a comprehensive survey of classic deconvolution methods for image restoration and reconstruction we refer the interested readers to \cite{banham_digital_1997} and \cite{puetter_digital_2005}.

In recent years there has been an increasing interest in the application of CS in various imaging modalities including but not limited to holography \cite{brady_compressive_2009}, coded aperture spectral imaging \cite{arce_compressive_2014}, fluorescent microscopy \cite{studer_compressive_2012}, and sub-wavelength imaging \cite{gazit_super-resolution_2009}. The CS-based imaging systems are particularly interesting in applications where the measurements are time-consuming or expensive. Furthermore, by exploiting the sparsity of the scene image, the CS imaging methods can operate at SNR regimes where conventional imaging methods may perform poorly.  A survey of practical advantages and challenges of various CS imaging systems can be found in \cite{compressed_willet_2011}. The first CS imaging system was introduced as the ``single-pixel camera'' in  \cite{duarte_single-pixel_2008} where a single sensor integrates the randomly masked versions of the scene image for a few different masks. Effectively, the single-pixel camera measures the inner product of the scene image and the randomly generated masks. Using the fact that natural images are often (nearly) sparse in some basis, it is shown in \cite{duarte_single-pixel_2008} that CS allows accurate image reconstruction in this single-pixel architecture.

In this paper, we consider the deconvolution problem in a random mask imaging system that is very similar to the single-pixel camera; the main difference is that the reflections from the digital micromirror device (DMD) are blurred in a controlled manner in order to allow sampling with a few sensors. Our goals is to determine how the performance of the system depends on the number of masks and the extent of the blurring and the subsampling. Based on an idealized mathematical model of the considered imaging system, we tie the number of masks sufficient for reconstruction of the scene image to certain characteristics of the blurring kernel and the subsampling operator. Because the integration that occurs at each sensor involves the convolution with the Point Spread Function (PSF) of the lens, our analysis has similarity with the analyses used for CS with random convolution \cite{romberg_compressive_2009, rauhut_2010_compressive, rauhut_2012_restricted}. These analyses, address the problems where convolution with a known random signal/filter is of interest. In our problem, however, the convolution is deterministic and randomness occurs as spatial phase modulation. Our theoretical results are based on recent theoretical developments regarding random matrices and chaos processes.

\subsection{Organization of the paper}
Section \ref{sec:ProblemSetup} elaborates on the considered imaging architecture and the (idealized) formulation of the deconvolution problems. The main theoretical guarantees are stated in Section \ref{sec:MainResults}. In Section~\ref{ssec:DLS}, we study the stability of the deconvolution for generic scene images by analyzing the behavior of the extreme singular values of the measurement matrix in terms of certain properties of the known blurring kernel and subsampling operator. Furthermore, in Section~\ref{ssec:SDL1}, the performance guarantees of $\ell_1$-minimization for recovering sparse images are stated in terms of the Restricted Isometry Property (RIP) of the measurement matrix. The proofs of these guarantees are provided in the appendices. The theoretical guarantees are validated by the  numerical simulations reported in Section \ref{sec:Experiments}. The concluding remarks are provided in Section \ref{sec:Conclusion}.
\subsection{Notation}
We use the following notation convention throughout this paper. Matrices and vectors are denoted by bold capital and small letters, respectively. The vectors $\mb{e}_i$ for $i=1,2,\dotsc$ denote the canonical basis vectors that are zero except at their $i$-th entry which is one. The diagonal matrix whose diagonal entries form a vector $\mb{x}$ is denoted by $\mb{D}_\mb{x}$. The matrix of diagonal entries of a matrix $\mb{X}$ is denoted by $\mr{diag}\left( \mb{X}\right)$. The vector norms $\left\Vert \cdot\right\Vert_p$ for $p\geq 1$ are the standard $\ell_p$-norms. The so called $\ell_0$-norm, which counts the nonzero entries of its argument, is denoted by $\left\Vert \cdot\right\Vert_0$. The matrix norms $\left\Vert\cdot\right\Vert$ and $\left\Vert\cdot\right\Vert_F$ are used to denote the operator norm and the Frobenius norm, respectively. The largest (or the smallest) eigenvalue of symmetric matrices are denoted by $\lambda_{\max}\left(\cdot\right)$ (or $\lambda_{\min}\left(\cdot\right)$). Also, we use $\mr{cond}\left(\cdot\right)$ to denote the condition number of matrices. Occasionally, expressions of the form $f\gtrsim g$ (or $f\lesssim g$) are used that should be interpreted as $ f \geq cg$ (or $c f\leq g$), where $c>0$ is some absolute constant.

\section{Problem Setup}
\label{sec:ProblemSetup}

\begin{figure*}
\centering
\includegraphics[width=0.75\textwidth]{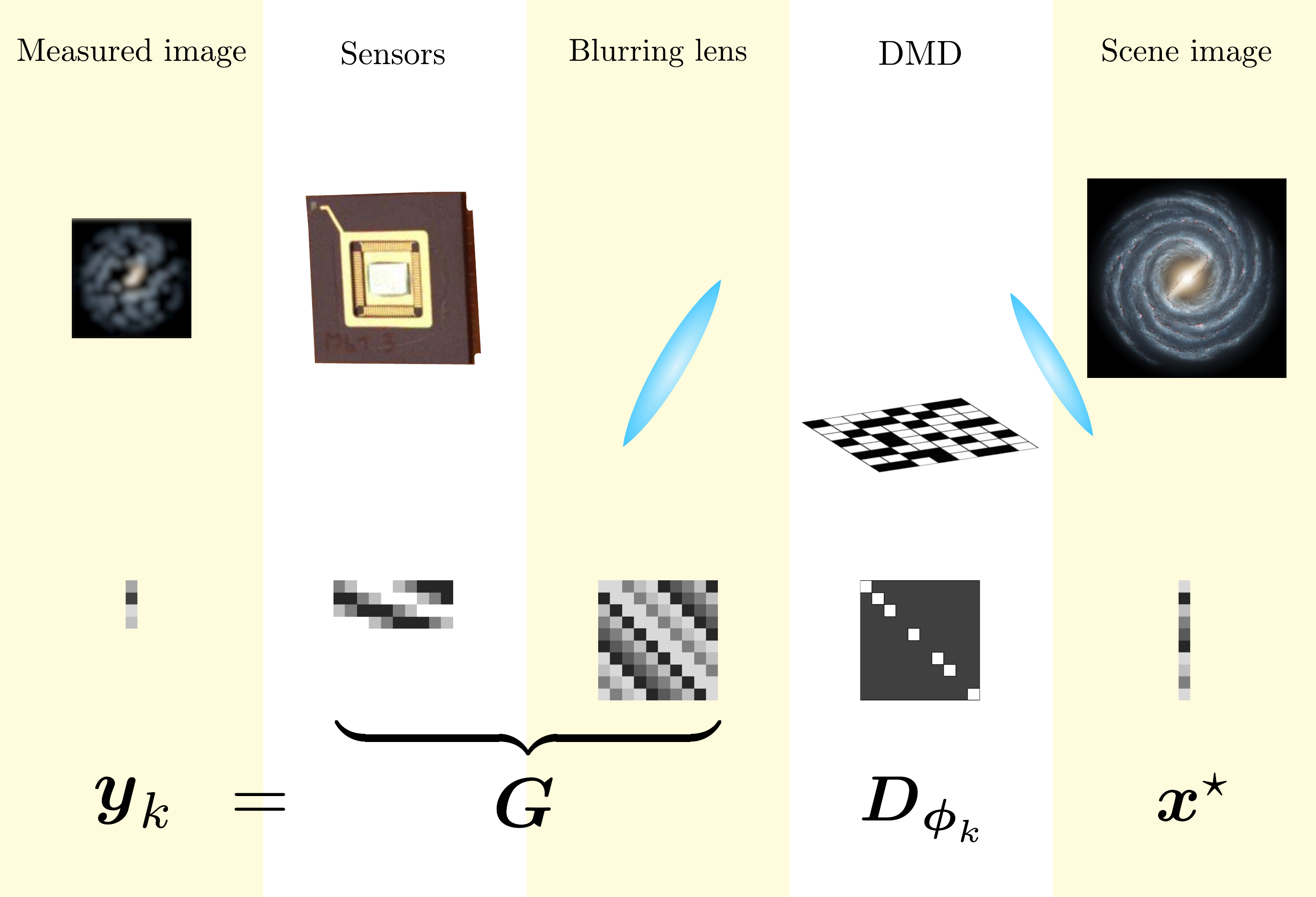}
\protect\caption{\label{fig:SchemaMI}Schematic of the masked imaging system. The scene image is focused on a DMD that acts as a spatial phase modulator. The masked image reflected from the DMD is blurred by a lens and then spatially subsampled by a few sensors. The result is one set of measurements corresponding to the chosen DMD pattern.}
\end{figure*}

The purpose of this paper is to develop a principled understanding of the type of imaging architecture depicted in Figure~\ref{fig:SchemaMI}.  The central abstraction we make is to model the acquisition process as the application of a matrix to an unknown vector.  This puts the image reconstruction problem squarely into the realm of linear algebra, allowing us to connect it to recent developed mathematics in that field.  In this section, we detail how this abstraction is made, and the algorithms we use to perform the reconstruction.

The main physical assumption we make is that the mapping between the true image $I(\mb{t})$ and a single measurement is a linear functional; each measurement can be written as
\[
	y_m = \langle I(\mb{t}), h_m(\mb{t})\rangle = \int I(\mb{t}) h_m(\mb{t})~d\mb{t}.
\]
The measurement ``test functions'' $h_m(\mb{t})$, which like the image are functions of a continuous 2D spatial index, are different for each sensor location and each pattern on the SLM.  They also depend on the point spread function for the blurring lens, and the size and shape of the sensor.  Our results will depend on general properties of the ensemble of these functions, but we do not assume that they have any particular form.  We will assume throughout that the $h_m(\mb{t})$ are known.  

We model the action of the spatial light modulator as follows: the input image is divided into small square regions, and a weight is applied uniformly over each region.  In the schematic above, the SLM is depicted as a DMD, implying that our weights are binary values, and are either $0$ or $1$.  The mathematical analysis in the next section will use $\pm 1$ for the binary weights; the analysis is smoother with weights that are zero-mean.  But if we have measurements using 0/1 weights $\mb{\phi}$, we can use the linearity of the system to easily generate the corresponding measurement $2\mb{\phi}-1$, which has entries of $\pm 1$.  All that is needed to perform this transformation is a single measurement with all of the weights equal to $1$. While the the choice of zero-mean random masks have some importance for the purpose of theoretical analysis, the functionality of the imaging system and the reconstruction algorithms are independent of this choice. We support this notion by a simulation on a 2D image with binary (i.e., $0/1$) masks in Section \ref{sec:Experiments}.

As with any inverse problem whose solution we actually want to compute, we model the underlying image $I(\mb{t})$ as lying in a finite dimensional subspace spanned by a set of known basis functions $\psi_\ell(\mb{t}), \ell=1,\ldots,L$.  We can write
\[
	I(\mb{t}) = \sum_{\ell=1}^L x_\ell^\star \psi_\ell(\mb{t}).
\]
To reconstruct $I$, we estimate the vector of expansion coefficients $\mb{x}\in\mbb{R}^L$.  With the basis model in place, 
we can now write the entire set of $M$ measurements $\mb{y}$ as an $M\times L$ matrix $\mb{H}$ applied to $\mb{x}$:
\[
	\mb{y} = \mb{H}\mb{x}^\star + \text{error},
\]
where
\[
	H(m,\ell) = \langle \psi_\ell(\mb{t}),h_m(\mb{t})\rangle = \int h_m(\mb{t})\psi_\ell(\mb{t})~d\mb{t}.
\]
The error term above is meant to encapsulate errors from all sources, including modeling inaccuracies (the underlying image does not truly lie in the span of the $\psi_\ell$) and the presence of noise in the measurements.

Our analysis uses a standard discretization: the $\psi_\ell(\mb{t})$ are indicator functions on the same small square regions over which the SLM divides the image.  As we describe further below, this allows us to write the action of the SLM on the image expansion coefficients as multiplication by a diagonal matrix.  This opens a path for the mathematical analysis of the architecture.  However, this not necessarily the only basis for which might be used for discretization; the algorithms used for the reconstruction only require us to provide the matrix $\mb{H}$.

For a given pattern on the SLM, whose weights are entries in the vector $\mb{\phi}_{k}$, we divide the measurement operator into two parts: an $L\times L$ diagonal matrix $\mb{D}_{\mb{\phi}_{k}}$ which maps the $\{x(i)\}$ to $\{\phi_k(i)x(i)\}$, and an $N\times L$ matrix $\mb{G}$, which maps the modulated coefficients into the values measured at the sensor array:
\[
	\mb{y}_k = \mb{G}\mb{D}_{\mb{\phi}_k}\mb{x}^\star + \mr{error}. 
\]
The measurements can be stacked and written compactly as
\begin{align*}
\mb{y} &=  \left[
						\begin{array}{c}
							\mb{y}_1 \\
							\mb{y}_2 \\
							\vdots \\
							\mb{y}_K
						\end{array} 						 
				\right] = \mb{H}\mb{x}^\star+\mr{error},
\end{align*}
where the measurement matrix $\mb{H}$ is given by
\begin{align}
\mb{H} & =\left[
												\begin{array}{c}
													\mb{G}\mb{D}_{\mb{\phi}_{1}} \\
													\mb{G}\mb{D}_{\mb{\phi}_{2}} \\
													 \vdots \\ 
													 \mb{G}\mb{D}_{\mb{\phi}_{K}}
												\end{array}
												\right]. \label{eq:mmx}
\end{align}
The analysis below draws the entries in each $\mb{\phi}_k$ independently and taking values $\pm 1$ with equal probability.  The number of patterns used is $K$, for a total of $M=KN$ measurements of the image.

The matrix $\mb{G}$ models the joint action of the blurring lens and the sensor.  If, for example, we model the sensors as taking point samples, then the $n$th row of $\mb{G}$ will contain the inner products $\langle \psi_\ell(\mb{t}), p(\mb{\tau}_n-\mb{t})\rangle$, where $p(\mb{t})$ is the point spread function of the lens, and $\mb{\tau}_n$ is the sample location for sensor $n$.  Alternatively, if we model the sensors as integrating over a certain region, we replace $p(\cdot)$ in the previous expression with its convolution with an indicator function over this region.  As mentioned before, the analysis below does not depend on the particular physical model that we use, but rather on general properties of $\mb{G}$.  Qualitatively, these amount to $\mb{G}$ not having any blind spots --- each part of the image influences the reading on at least one of the sensors --- and each sensor having a distinct view of the image.
 
In this paper, we will take $\mb{G}$ to be known, making the recovery of $\mb{x}^\star$ a certain kind of generalized deconvolution problem.  We analyze the performances of two algorithms, one of which assumes no structure in the image, and the other tailored to the case where the image is sparse.
\begin{enumerate}
\item \emph{Deconvolution of generic scene images}: In this scenario, the aim is to estimate the scene image $\mb{x}^\star$ that has no specific structure. For this deconvolution problem we analyze the least squares estimator
\begin{align}
\widehat{\mb{x}} & :=\argmin_{\mb{x}}\left\| \mb{H}\mb{x}-\mb{y}\right\|_2^2 =  \sum_{k=1}^{K}\left\Vert \mb{G}\mb{D}_{\mb{\phi}_{k}}\mb{x}-\mb{y}_{k}\right\Vert _{2}^{2}.\label{eq:ls}
\end{align}
The performance of the above least squares depends solely on the conditioning of the cumulative measurement matrix $\mb{H}$ in \eqref{eq:mmx}.
This is a matrix with structured randomness (since the $\mb{\phi}_k$ are random); we will show that it is well-conditioned if the total number of measurements $KN$ slightly larger than then number of pixels $L$ we use to discretize the image.

\item \emph{Deconvolution of sparse scene images}: We also study the more interesting problem of recovering sparse images in the described imaging system. With the measurement matrix $\mb{H}$ defined in \eqref{eq:mmx}, we consider the estimator
\begin{align}
\widehat{\mb{x}} & :=\argmin_{\mb{x}}\left\Vert \mb{x}\right\Vert _{1}\label{eq:l1-min}\\
 & \hspace{1.5em}\text{subject to }\mb{H}\mb{x}=\mb{y},\nonumber 
\end{align}
 which is inspired by the CS framework. Note that we consider only error-free measurements; the extension to the case of corrupted measurements is the same as in well-known CS approaches. A common sufficient condition used in CS to guarantee accuracy of not only the $\ell_1$-minimization algorithm, but also a variety of greedy algorithms, is the Restricted Isometry Property (RIP) (see \cite{foucart_sparse_2012} and references therein). A matrix $\mb{A}$ is said to satisfy the RIP with constant $\delta_S\in\left(0,1\right)$ if 
 \begin{align*}
\left(1-\delta_S\right) \left\Vert \mb{x}\right\Vert^2_2 & \leq  \left\Vert \mb{A}\mb{x}\right\Vert^2_2 \leq \left(1+\delta_S\right)\left\Vert \mb{x}\right\Vert^2_2,
\end{align*}
holds for all $S$-sparse vectors $\mb{x}$. Exact recovery of $S$-sparse signals via $\ell_1$-minimization is shown in \cite{candes_restricted_2008} under the condition $\delta_{2S}<\sqrt{2}-1$, which was later improved to $\delta_{2S}<3/\left(4+\sqrt{6}\right)$ in \cite{foucart_note_2010}.
 
 We provide a sufficient condition on the number of masks that can guarantee RIP for the measurement matrix \eqref{eq:mmx}. Our  goal is to show that the sparse scene image can be recovered through \eqref{eq:l1-min}, even if the number of masks grows much slower than the dimension of the target image at a rate that is almost linear in its sparsity. The constant factor in the rate is determined by the relative sensitivity of the measurements to different pixels of the scene image.
\end{enumerate}

\section{Main Results}\label{sec:MainResults}
To state the main results we use the shorthand
\begin{align*}
\rho  &:=\left\Vert \mb{G}\right\Vert,&
\theta_{\max}  &:=\max_{1\leq i\leq L}\left\Vert \mb{G}\mb{e}_{i}\right\Vert _{2}, &
\theta_{\min}  &:=\min_{1\leq i\leq L}\left\Vert \mb{G}\mb{e}_{i}\right\Vert _{2},
\end{align*}
respectively for the spectral norm, the largest column $\ell_2$-norm, and the smallest column $\ell_2$-norm of $\mb{G}$.

 Intuitively, to ensure identifiability for every scene image in the imaging system described in Section \ref{sec:ProblemSetup}, the measurements need to retain the relative intensity of the scene pixels to a great extent. For instance, no scene image with only a single active pixel should be mapped to an all-zero measurement (i.e., $\theta_{\min}>0$). Similarly, there should be no drastic amplification of any pixel with respect to other pixels. These requirements translate into desirability of matrices $\mb{G}$ whose columns $\ell_2$-norms are almost equal (i.e. $\tfrac{\theta_{\max}}{\theta_{\min}}\approx 1$). In fact, straightforward calculations reveal that the expectation of the matrix $\mb{H}^\mr{T}\mb{H}$, which describes the system, has a condition number equal to $\left(\tfrac{\theta_{\max}}{\theta_{\min}}\right)^2$.
 
 \subsection{Deconvolution by Least Squares}\label{ssec:DLS}
 As discussed above, identifiability of the scene image in the considered imaging system depends on certain characteristics of the matrix $\mb{H}^\mr{T}\mb{H}$ where $\mb{H}$ is given by \eqref{eq:mmx}. The conditioning of the same matrix determines the stability of the estimate obtained by the least squares estimator \eqref{eq:ls}. To verify this fact, it suffices to observe that the least squares solution can be written as 
 \begin{align*}
	 \widehat{\mb{x}}&=\left(\mb{H}^\mr{T}\mb{H}\right)^{-1}\mb{H}^\mr{T}\left(\mb{y}^\star+\mb{z}\right)=\mb{x}^\star+\left(\mb{H}^\mr{T}\mb{H}\right)^{-1}\mb{H}^\mr{T}  \mb{z},
 \end{align*}
 where $\mb{y}^\star=\mb{H}\mb{x}^\star$ denotes the vector of error-free measurements and $\mb{z}$ is a measurement perturbation. It is clear from the equation above that the conditioning of $\left(\mb{H}^\mr{T}\mb{H}\right)^{-1}\mb{H}^\mr{T}$, or equivalently that of $\mb{H}^\mr{T}\mb{H}$, determines the stability against the perturbation. The following theorem establishes a relation between the number of
applied masks (i.e., $K$) and the condition number of $\mb{H}^\mr{T}\mb{H}$.
\begin{theorem}[Conditioning of the least squares for deconvolution]
\label{thm:LS} Suppose that 
\begin{align}
K & \geq\frac{\beta+1}{\log 4 -1}\delta^{-2}\frac{\rho^{2}}{\theta_{\min}^{2}}\log L,\label{eq:LS-bound}
\end{align}
 for absolute constants $\delta\in\left[0,1\right]$ and $\beta>0$.
Then with probability at least $1-2L^{-\beta}$ we have 
\begin{align*}
\mr{cond}\left(\mb{H}^{\mr{T}}\mb{H}\right) & \leq\frac{1+\delta}{1-\delta}\cdot\frac{\theta_{\max}^{2}}{\theta_{\min}^{2}}.
\end{align*}

\end{theorem}

\begin{remark}
\rm
As the proof provided in the appendix shows, the bound in \eqref{eq:LS-bound} can be improved slightly to
 \begin{align*}
 K&\geq\left(\beta+1\right) \max\left\{ \frac{\rho^{2}}{\psi\left(-\delta\right)\theta_{\min}^{2}},\frac{\rho^{2}}{\psi\left(\delta\right)\theta_{\max}^{2}}\right\},
 \end{align*} 
 where $\psi\left(t\right):=\left(1+t\right)\log\left(1+t\right)-t$. The bound \eqref{eq:LS-bound} is preferred merely because of its simpler expression. 
 \end{remark}
 \begin{remark}
 \rm
 It is worthwhile to inspect the tightness of the bound imposed by \eqref{eq:LS-bound} qualitatively. We can express the term $\rho^2/\theta^2_{\min}$ in \eqref{eq:LS-bound} as the product of $\left\Vert\mb{G}\right\Vert^2_F/\theta^2_{\min}$ and $\rho^2/\left\Vert\mb{G}\right\Vert^2_F$. As discussed above,  for a fixed number of masks the more uneven the column norms of $\mb{G}$ are, the more unbalanced the pixel amplification/attenuation and thereby the more unstable the least squares become. The first term (i.e., $\left\Vert\mb{G}\right\Vert^2_F/\theta^2_{\min}$) that varies in the interval $\left[L,\infty\right)$, measures how equally the energy of $\mb{G}$ is distributed among its columns. Another crucial factor that determines the required number of masks is the redundancy of the measurements that can be captured by the rank of the matrix $\mb{G}$. A severely rank-deficient $\mb{G}$ provides less information per mask than a full-rank $\mb{G}$. The second term mentioned above (i.e., $\rho^2/\left\Vert\mb{G}\right\Vert^2_F$) that varies in the interval $\left[1/\mr{rank}\left(\mb{G}\right),1\right]$ quantifies the dependence on the rank of $\mb{G}$. In the ideal case of a full-rank $\mb{G}$ with equally normed columns \eqref{eq:LS-bound} imposes $K\gtrsim \frac{L\log L}{N}$ which is suboptimal merely by a factor of $\log L$. In the case that $\mb{G}$ has columns of equal norm, but $\mr{rank}\left(\mb{G}\right)=1$, \eqref{eq:LS-bound} requires $K\gtrsim L\log L$ that is also suboptimal by a mere factor of $\log L$. 
\end{remark}

\subsection{Sparse Deconvolution by $\ell_1$-Minimization}\label{ssec:SDL1}
It is also desirable to ensure that the $\ell_1$-minimization \eqref{eq:l1-min} is stable and its performance degrades gracefully with perturbations. Significant shrinkage or expansion of the distances between different sparse signals under the measurement matrix $\mb{H}$ can lead to ambiguity or noise amplification, respectively. The RIP provides a formal characterization of the measurement matrices that have the desired properties. The next theorem establishes the RIP for a properly scaled version of the matrix $\mb{H}$. Below, the set of $S$-sparse vectors on the $L$-dimensional unit sphere is denoted by 
\begin{align*}\mc{D}_{S,L}&:=\left\{ \mb{x}\in\mbb{R}^{L}\mid\left\Vert \mb{x}\right\Vert _{2}=1\text{ and }\left\Vert \mb{x}\right\Vert _{0}\leq S\right\} .
\end{align*}

\begin{theorem}[RIP for sparse deconvolution]
\label{thm:SparseLS} Let $\mu:=\frac{\theta^{2}_{\max}}{\theta^{2}_{\min}}\geq 1$, $\theta^2_\mr{avg}:=\frac{\theta^2_{\max}+\theta^2_{\min}}{2}$, and suppose that
\begin{align*}
K & \gtrsim\delta^{-2}\mu^2 S\log L,
\end{align*}
 for some parameter $\delta\in\left(0,1\right)$. Then 
\begin{alignat}{2}
1-\frac{\mu-1+2\delta}{\mu+1} & \leq\frac{1}{\theta^2_{\mr{avg}}}\cdot\frac{\left\Vert \mb{H}\mb{x}\right\Vert _{2}^{2}}{K} & \leq 1+\frac{\mu-1+2\delta}{\mu+1}\label{eq:SLS-bound}
\end{alignat}
holds for all $\mb{x}\in\mc{D}_{S,L} $
with probability at least $1-2\mr{e}^{-CK\min\left\{\delta/\left(2\mu\right),\delta^{2}/\left(2\mu\right)^2\right\}}$,
where $C>0$ is an absolute constant.\end{theorem}
\begin{corollary}
Let the true signal $\mb{x}^\star$ be $S$-sparse. If $\mu$ is sufficiently close to one and $K\gtrsim\mu^2S\log L$,
then there exist $\delta\in\left(0,1\right)$ for which, \eqref{eq:l1-min}
recovers $\mb{x}^{\star}$ with probability
at least $1-2\mr{e}^{-CK\min\left\{ \delta/\left(2\mu\right),\delta^{2}/\left(2\mu\right)^2\right\}}$,
where $C>0$ is an absolute constant.
\end{corollary}
\begin{proof}
It is known from the standard compressed sensing literature that if
a measurement matrix satisfies the RIP of order $2S$ with a sufficiently
small constant, then the associated $\ell_{1}$-minimization exactly
recovers any $S$-sparse signal \cite[and references therein]{foucart_sparse_2012}. Note that \eqref{eq:SLS-bound}
guarantees that if $K\gtrsim\delta^{-2}\mu^2S\log L$,
then with probability at least \mbox{$1-2\mr{e}^{-CK\min\left\{\delta/\left(2\mu\right),\delta^2/\left(2\mu\right)^2\right\}}$} the matrix $\frac{1}{\sqrt{K}\theta_{\mr{avg}}}\mb{H}$
obeys the RIP of order $2S$
with a constant no more than $\left(\mu+2\delta-1\right)/\left(\mu+1\right)=1-2\left(1-\delta\right)/\left(\mu+1\right)$.
 For a $\mu$ that is sufficiently close to one
we can choose a $\delta\in\left(0,1\right)$ for which the desired
RIP constant can be achieved and thus the  $\ell_{1}$-minimization \eqref{eq:l1-min} successfully recovers $\mb{x}^{\star}$.
\end{proof}
\begin{remark}
\rm
Robustness of the proposed sparse deconvolution to noise can also be established using the standard RIP-based bounds known for $\ell_1$-minimization \cite[and references therein]{foucart_sparse_2012}, but we do not reproduce them here. Note that, as can be expected, for large values of $\mu$ the RIP and thereby robustness cannot be guaranteed.
\end{remark}

\section{Numerical Experiments}\label{sec:Experiments}
In this section, we provide numerical experiments in an effort to understand the empirical relationship between the number of masks used and the success of our deconvolution methods.

In Section~\ref{ssec:1DSim} below, we describe the aggregate results of a large number of small numerical experiments.  We present empirical phase transition diagrams of both the generic and the sparse deconvolution methods through a series of Monte Carlo simulations on one-dimensional signals.
The purpose here is not necessarily to mimic a real imaging system as close as possible, but rather to get a feel for how well random matrices of the general form \eqref{eq:mmx} work for least squares and sparse recovery.

In Section~\ref{ssec:2DSim}, we describe a single simulation that more closely matches a physical acquistion.  We use a realistic point spread function for the blurring operation, and $\{0,1\}$ weights for the spatial light modulator.  The image is also not exactly sparse.  We note that the recovery is successful despite the slight departure from what was analyzed theoretically.

\subsection{Deconvolution for 1D signals}\label{ssec:1DSim}
We evaluated the empirical success rate of the proposed generic and sparse deconvolution methods on synthetic 1D signals in relation to the number of masks $K$. In these simulations the subsampling operator, which models $N = 4$ sensors with uniform spacing, is fixed during the trials. Furthermore, the simulated blurring is modeled by a random circulant matrix whose first column $\mb{p}$ is a random blurring kernel drawn independently in each trial. Therefore, the matrix $\mb{G}$ consists of $N = 4$ uniformly spaced rows of the circulant matrix generated by $\mb{p}$. Also, as mentioned in Section II the masks are drawn independently each with independent equiprobable $\pm 1$ entries.

\subsubsection{Deconvolution by Least Squares}
We considered signals of dimension $L=2048$ and varied the number of masks from $K=512$ to $K=2048$ with the steps of size $32$. For the (blurring) filter (i.e., $\mb{h}$) we consider an all-pass model and a low-pass model. For the all-pass model the filter is drawn from a $\mc{N}\left(\mb{0},\mb{I}\right)$ distribution. For the low-pass model, however, the filter is obtained by suppressing the DFT coefficients of a standard normal random vector whose indices are between $128$ and $1920=2047-127$ leading to a random filter of ``bandwidth'' $127$. Given that the stability of the least squares approach can be characterized by the conditioning of the associated measurement matrix, we merely measured the ratio of the smallest and largest singular values of the matrix of interest. For each value of $K$ we ran the simulation $100$ times and computed the $10\%$, the $50\%$, and $90\%$ quantiles of the condition number.  

Figures \ref{fig:LS-a} and \ref{fig:LS-b} show the condition number of the matrix associated with the least squares problem \eqref{eq:ls} in logarithmic scale versus the number of masks $K$ for the considered all-pass and low-pass models, respectively. The first few values of $K$ are cropped out because the matrix was severely ill-conditioned and the condition number reached values in the order of $10^4$. As expected, increasing the number of masks improves the considered condition number. Furthermore, for any given number of masks the condition number for the all-pass model is better that the low-pass model by a factor of about two. Comparison between the quantiles of the condition number at any given value of $K$ also suggests that the condition number has more fluctuations in the case of the low-pass model.

\begin{figure*}
\noindent
\begin{subfigure}[t]{0.5\textwidth}
\centering
\includegraphics[width=1\textwidth]{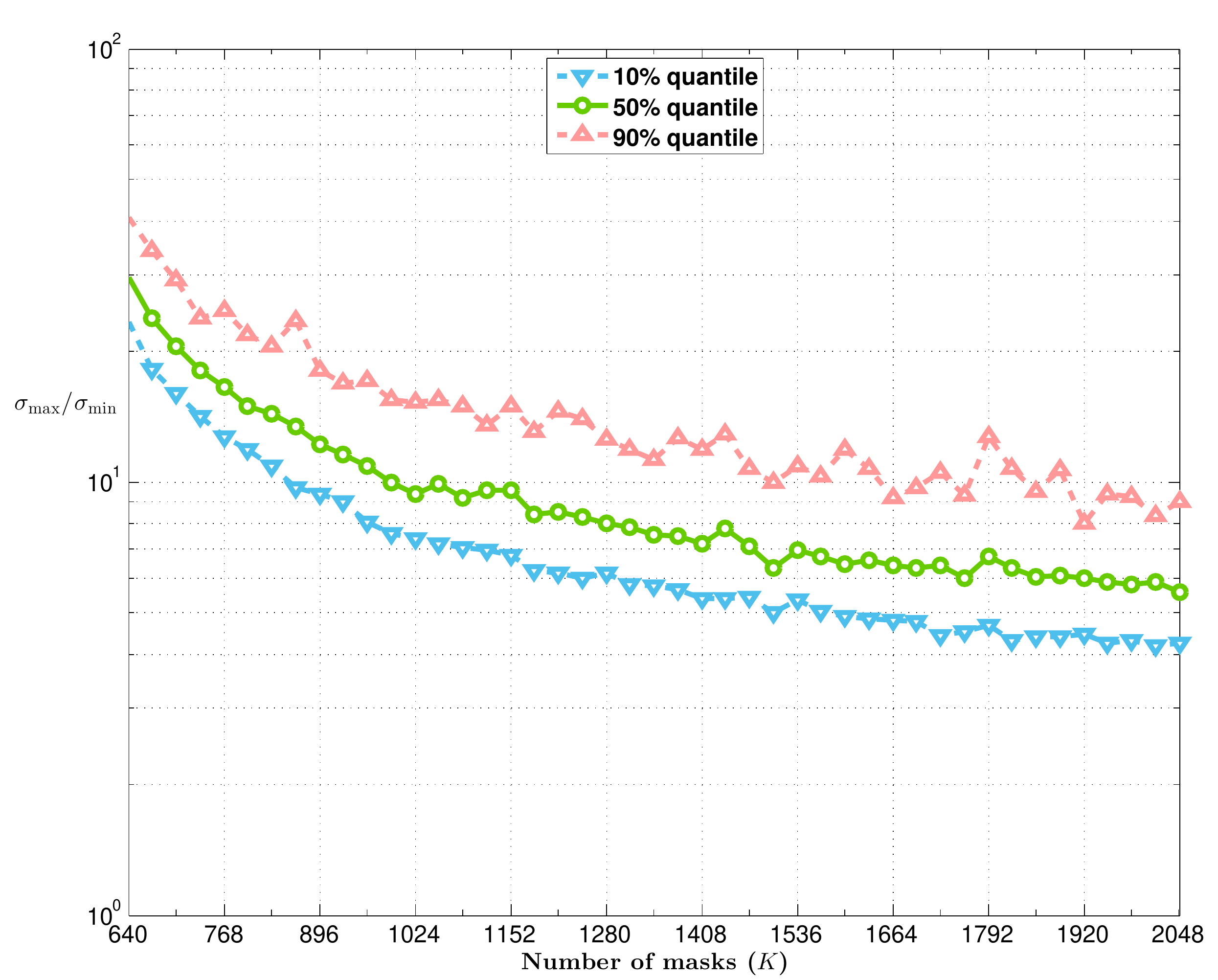}
\protect\caption{All-pass model}
\label{fig:LS-a}
\end{subfigure}
\begin{subfigure}[t]{0.5\textwidth}
\centering
\includegraphics[width=1\textwidth]{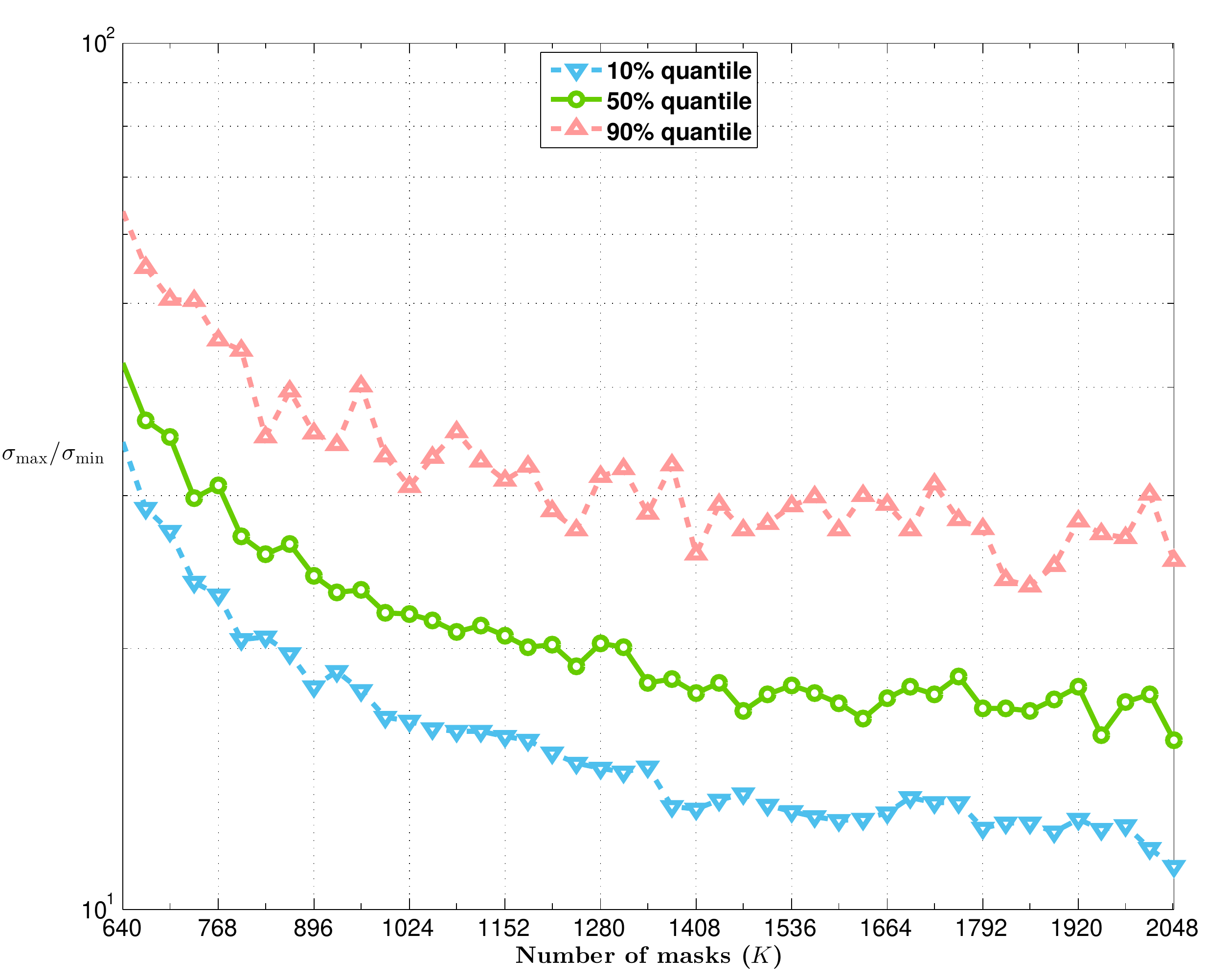}
\protect\caption{Low-pass model with a bandwidth of  127}
\label{fig:LS-b}
\end{subfigure}
\protect\caption{Conditioning of the measurement matrix in the least squares problem \eqref{eq:ls} vs. $K$, the number of masks.}
\label{fig:LS}
\end{figure*}

\subsubsection{Sparse Deconvolution by $\ell_1$-Minimization}
For this set of simulations we considered signals of dimension $L=2048$ whose nonzero entries are drawn independently from a standard normal distribution and are located on a support set drawn uniformly at random. We varied the number of masks from $K=8$ to $K=256$ with the steps of size eight. The (blurring) filter is generated according to the random all-pass and low-pass models described above for the case of unstructured deconvolution. These models might not be reasonable in incoherent imaging where the signal and the filter are both nonnegative. However, our theoretical guarantees similar to other results in CS suggests no significant dependence on the sign of the signal that can affect the simulation results.  For each value of $K$ the sparsity of the signal is increased as a multiple of eight until the successful recovery rate drops below $1\%$. Among each of the $100$ trials ran for each pair of $K$ and $S$ (i.e., the $\ell_0$-norm of the signal), those that yield a relative error no more than $5\%$  are counted as successful reconstructions. 

Figures \ref{fig:SLS-a} and \ref{fig:SLS-b} show the phase transition diagram of the estimate produced by \eqref{eq:l1-min} in terms of the number of masks $K$ and sparsity level $S$ respectively for the considered all-pass and low-pass models.  As can be seen from the figures phase transition boundary is almost linear in both models,  which is in agreement with the relation between $K$ and $S$ suggested by the Theorem \ref{thm:SparseLS}. While the phase transition boundary for the low-pass model is slightly lower (worse) than that of the all-pass model, the difference does not appear to be significant.

\begin{figure*}
\noindent
\begin{subfigure}[t]{0.5\textwidth}
\centering
\includegraphics[width=1\textwidth]{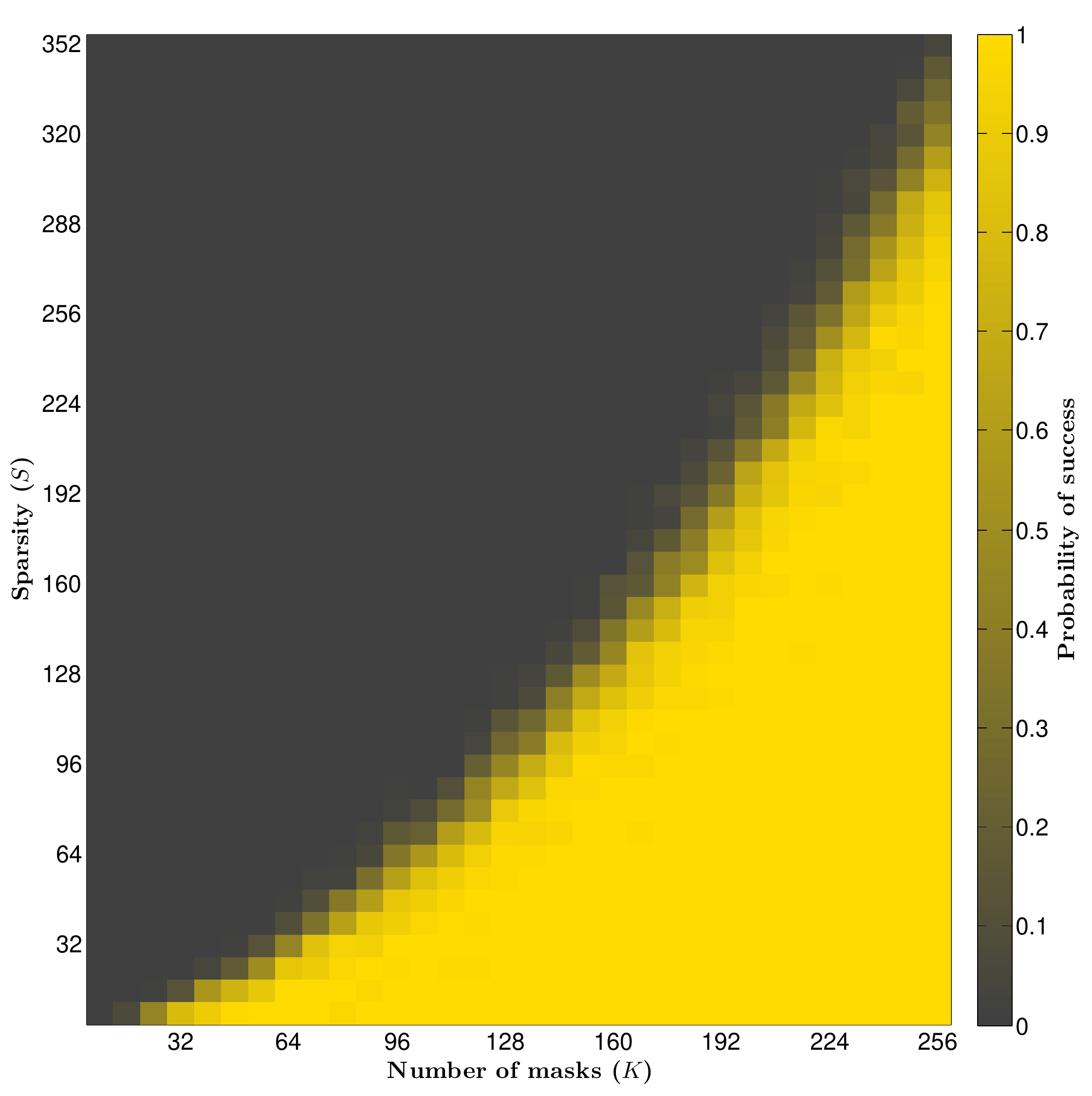}
\protect\caption{All-pass model}
\label{fig:SLS-a}
\end{subfigure}
\begin{subfigure}[t]{0.5\textwidth}
\centering
\includegraphics[width=1\textwidth]{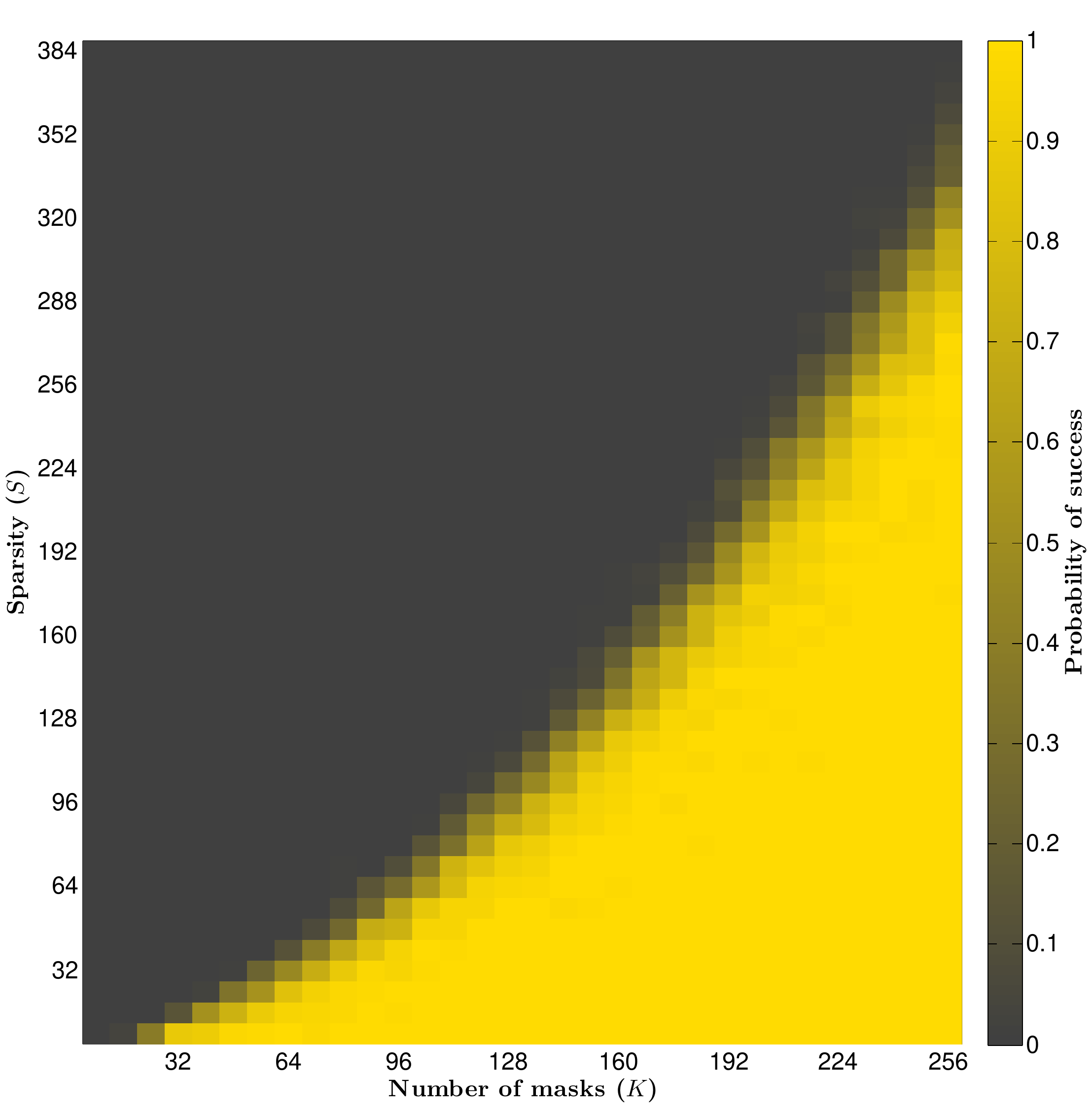}
\protect\caption{Low-pass model}
\label{fig:SLS-b}
\end{subfigure}
\protect\caption{The phase transition diagram of $\ell_1$-minimization in terms of $K$, the number of masks, and $S$, the $\ell_0$-norm of the signals of dimension $L=2048$.}
\label{fig:SLS}
\end{figure*}

\subsection{2D Sparse Deconvolution}\label{ssec:2DSim}
We also applied the $\ell_1$-minimization technique to recover a 2D image from (synthetic) measurements in the considered masked imaging system and the result is depicted in Figure \ref{fig:SCS}. Instead of reconstruction in the canonical basis, in this experiment we considered reconstruction of the image with respect to the 2D-DCT basis. Furthermore, the random masks are populated with i.i.d. equiprobable $0/1$  entries rather than the $\pm1$ entries used in the previous simulations. The target image at the top left corner is a fluorescent microscopy image\footnote{The image is an adaptation of an image on Wikimedia Commons that is available online at:\\* \texttt{http://commons.wikimedia.org/wiki/File:S\_cerevisiae\_septins.jpg}} that has $L= 188\times 256$ pixels with $8$-bit per color channel. Although the image appears to be sparse, only about $80\%$ of its pixels are zero-valued. In the 2D-DCT basis, however, the image is more \emph{compressible} as, for instance, its squared $\ell_1$-norm to $\ell_2$-norm ratio is less than $4428\approx 0.09 L$. The blurring kernel, a $128\times 128$ PSF generated using the PSFGenerator package \cite{kirshner_psf_2013}, is depicted in the bottom left corner of Figure \ref{fig:SCS}. For comparison, the blurred unmasked image is shown at the top center. We applied  $K=50$ random $0/1$ masks and subsampled each of the blurred outputs by a rate of  $1/11$  in each direction which yields $N=18\times 24$ scalar measurements per mask. These $50$ measurements are shown at the center of the second row. The overall undersampling rate is $K\times N/L\approx 45\%$. We applied the $\ell_1$-minimization with the same blur kernel for each color channel, independently. The relative error of estimated image obtained by $\ell_1$-minimization, shown at the top right corner of Figure \ref{fig:SCS}, is less than $10\%$. We repeated the same simulation with $\pm1$ masks and the relative error and the quality of the estimate remained virtually the same.

\begin{figure*}
\noindent
\centering
\includegraphics[width=1\textwidth]{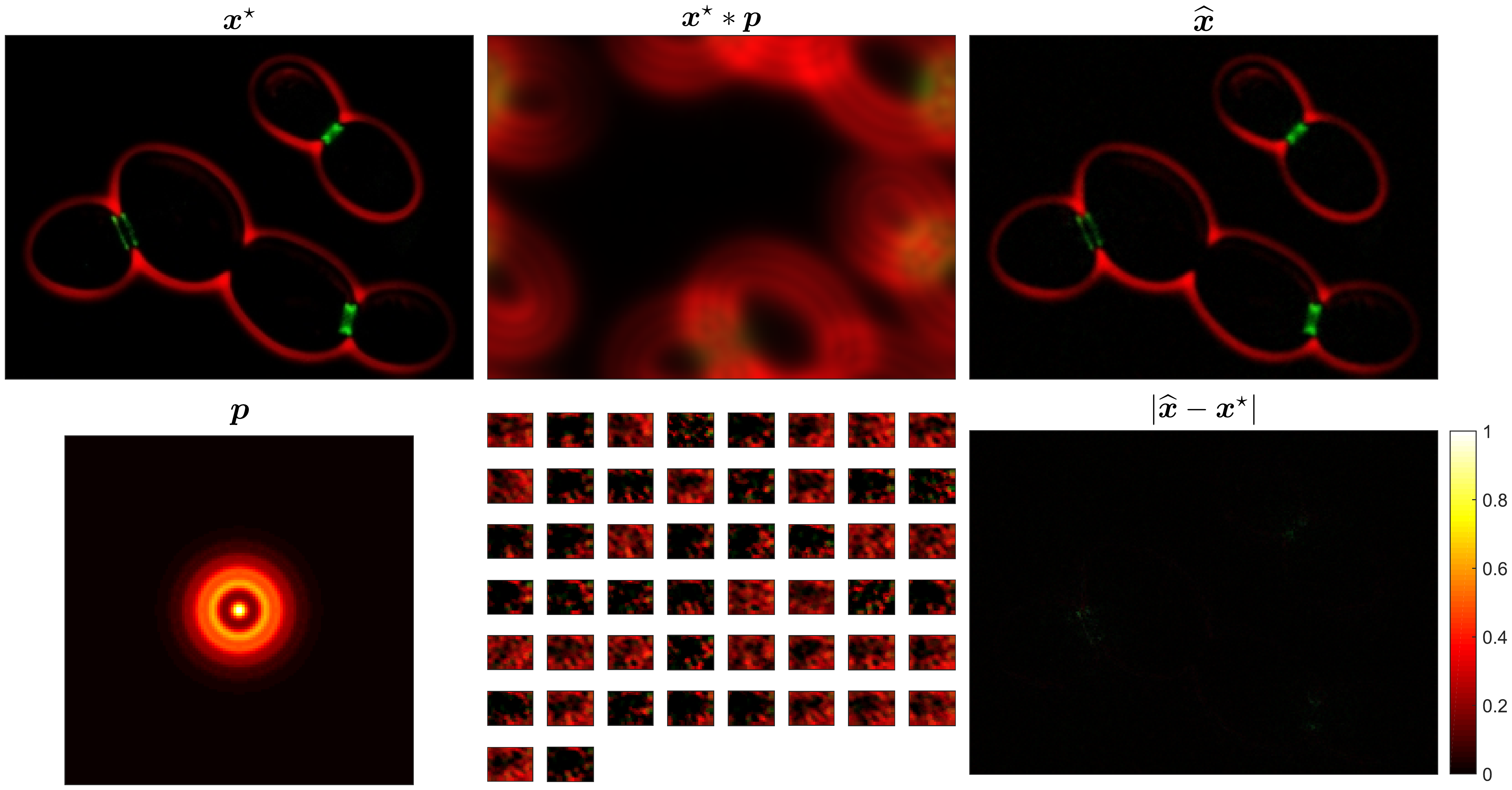}
\protect\caption{Sparse deconvolution applied to fluorescent micrograph of cell outlines (red) and \emph{septins} (green) in \emph{Saccharomyces}. The first column shows the original $188\times 256$ scene image (top) and the $128\times 128$ blur kernel (bottom). The second column shows the blurred image (top) for comparison and the $50$ masked, blurred, and subsampled measurements (bottom). The last column shows the estimated image (top) and the absolute error (bottom).}
\label{fig:SCS}
\end{figure*}

\section{Conclusion \label{sec:Conclusion}}
In this paper we studied the deconvolution problem in an idealized random mask imaging system. We quantified the number of masks that suffices to solve the deconvolution problem stably. For generic scene images and depending on the extent of the blurring and subsampling, we can have a well-conditioned deconvolution problem at the cost of oversampling at a rate logarithmic in the dimension of the target image. For sparse scene images, however, stable deconvolution through $\ell_1$-minimization is possible at much lower sampling rates. The number of required masks for stable sparse deconvolution can grow almost linearly in the sparsity at a rate much slower than the dimension of the target image. The established bounds can be satisfactory in certain regimes, but the sharpness of these bounds needs further investigation. For example, the sample complexity stated in Theorem \ref{thm:SparseLS} might be pessimistic as its dependence on the number of sensors is not ideal.

Considering more realistic mathematical models of the random mask imaging system introduces interesting and challenging problems that can be studied in the future. For example, analyzing the system under Poisson noise model would be of great interest as it is a more common and realistic model in optics. Furthermore, if the sensors only measure intensity, we would obtain a phase retrieval problem whose theoretical analysis may requires significantly different approaches.

\appendices
\section{Proof of Theorem \ref{thm:LS}}

To prove Theorem \ref{thm:LS}, we use a variant of matrix Chernoff
tail bounds for sums of random positive-semidefinite matrices due
to \cite[Corollary 5.2]{tropp_user-friendly_2011}.

\begin{IEEEproof}[Proof of Theorem \ref{thm:LS}.]
 Let \textbf{$\mb{Z}_{k}=\mb{D}_{\mb{\phi}_{k}}\mb{G}^{\mr{T}}\mb{G}\mb{D}_{\mb{\phi}_{k}}$}. Our goal is to show that for sufficiently large $K$, the random
matrix $\sum_{k=1}^{K}\mb{Z}_{k}=\mb{H}^{\mr{T}}\mb{H}$
is well-conditioned with high probability. It is straightforward to
verify that $\mb{Z}_{k}\succcurlyeq\mb{0}$ and $\left\Vert \mb{Z}_{k}\right\Vert =\rho^{2}$.
Furthermore, we have 
\begin{align*}
\mbb{E}\left[\sum_{k=1}^{K}\mb{Z}_{k}\right] & =K\mr{diag}\left(\mb{G}^{\mr{T}}\mb{G}\right),
\end{align*}
from which we can deduce that 
\begin{align*}
\lambda_{\min}\left(\mbb{E}\left[\sum_{k=1}^{K}\mb{Z}_{k}\right]\right) & =K\theta_{\min}^{2}
\end{align*}
 and 
\begin{align*}
\lambda_{\max}\left(\mbb{E}\left[\sum_{k=1}^{K}\mb{Z}_{k}\right]\right) & =K\theta_{\max}^{2}.
\end{align*}
 Therefore, with $\psi\left(\delta\right):=\left(1+\delta\right)\log\left(1+\delta\right)-\delta$, we can apply the matrix Chernoff bound \cite[Corollary 5.2]{tropp_user-friendly_2011}
and obtain 
\begin{align*}
\mbb{P}\left(\lambda_{\min}\left(\sum_{k=1}^{K}\mb{Z}_{k}\right)\leq\left(1-\delta\right)K\theta_{\min}^{2}\right) & \leq L\,e^{-\frac{K\theta_{\min}^{2}\psi\left(-\delta\right)}{\rho^{2}}} 
\end{align*}
for any $\delta\in\left[0,1\right]$ and 
\begin{align*}
\mbb{P}\left(\lambda_{\max}\left(\sum_{k=1}^{K}\mb{Z}_{k}\right)\geq\left(1+\delta\right)K\theta_{\max}^{2}\right) & \leq L\,e^{-\frac{K\theta_{\max}^{2}\psi\left(\delta\right)}{\rho^{2}}},
\end{align*}
for any $\delta\geq0$.
 In particular, using \eqref{eq:LS-bound} and the fact that $\psi\left(\pm\delta\right)\geq \left(\log 4 -1\right)\delta ^2$ for all $\delta\in \left(0,1\right)$, it is straightforward
to show that 
\begin{align*}
\mbb{P}\left(\lambda_{\min}\left(\sum_{k=1}^{K}\mb{Z}_{k}\right)\leq\left(1-\delta\right)K\theta_{\min}^{2}\right) & \leq L^{-\beta}
\end{align*}
 and 
\begin{align*}
\mbb{P}\left(\lambda_{\max}\left(\sum_{k=1}^{K}\mb{Z}_{k}\right)\geq\left(1+\delta\right)K\theta_{\max}^{2}\right) & \leq L^{-\beta}.
\end{align*}
Therefore, for values of $K$ that obey \eqref{eq:LS-bound} we have
\begin{align*}
\frac{\lambda_{\max}\left(\sum_{k=1}^{K}\mb{Z}_{k}\right)}{\lambda_{\min}\left(\sum_{k=1}^{K}\mb{Z}_{k}\right)} & \leq\frac{1+\delta}{1-\delta}\cdot\frac{\theta_{\max}^{2}}{\theta_{\min}^{2}},
\end{align*}
with probability at least $1-2L^{-\beta}$.
\end{IEEEproof}

\section{Proof of Theorem \ref{thm:SparseLS}}

Let $\mb{A}_{\mb{x}}:=\mb{G}\mb{D}_{\mb{x}}$
and 
\begin{align*}
\widetilde{\mb{A}}_{\mb{x}} & =\mb{I}_{K\times K}\otimes\mb{A}_{\mb{x}}=\left[\begin{array}{ccc}
\mb{A}_{\mb{x}} &  & \mb{0}\\
 & \ddots\\
\mb{0} &  & \mb{A}_{\mb{x}}
\end{array}\right].
\end{align*}
 Furthermore, define 
\begin{align}
\mc{A} & :=\left\{ \widetilde{\mb{A}}_{\mb{x}}\mid\mb{x}\in\mc{D}_{S,L}\right\} .\label{eq:Set}
\end{align}
To state the proof, we also need to define certain quantities as follows.
Let 
\begin{align}
d_{F}\left(\mc{A}\right) & :=\sup_{\mb{A}\in\mc{A}}\left\Vert \mb{A}\right\Vert _{F} 
 	=\sup_{\mb{x}\in\mc{D}_{S,L}}\left\Vert \widetilde{\mb{A}}_{\mb{x}}\right\Vert _{F}\nonumber \\
 & =\sqrt{K}\sup_{\mb{x}\in\mc{D}_{S,L}}\left\Vert \mb{A}_{\mb{x}}\right\Vert _{F}
   =\sqrt{K}\theta_{\max}\label{eq:FrobeniusNorm}
\end{align}
 and 
\begin{align}
d\left(\mc{A}\right) & :=\sup_{\mb{A}\in\mc{A}}\left\Vert \mb{A}\right\Vert 
  =\sup_{\mb{x}\in\mc{D}_{S,L}}\left\Vert \widetilde{\mb{A}}_{\mb{x}}\right\Vert \nonumber \\
 & =\sup_{\mb{x}\in\mc{D}_{S,L}}\left\Vert \mb{A}_{\mb{x}}\right\Vert 
  =\theta_{\max}.\label{eq:OperatorNorm}
\end{align}
Let $\gamma_{2}\left(\mc{A},\left\Vert \cdot\right\Vert \right)$
be the Talagrand's functional \cite{talagrand_generic_2005} for the set $\mc{A}$ with respect
to the operator norm. It is known that this functional can be bounded
from above by the Dudley's integral as 
\begin{align}
\gamma_{2}\left(\mc{A},\left\Vert \cdot\right\Vert \right) & \leq c_{1}\int_{0}^{d\left(\mc{A}\right)}\sqrt{\log N\left(\mc{A},\left\Vert \cdot\right\Vert ,u\right)}\mr{d}u,\label{eq:DudleyBound}
\end{align}
where $N\left(\mc{X},\left\Vert\cdot\right\Vert_\mc{X},\varepsilon\right)$ denotes the covering number of a set $\mc{X}$ with respect to $\varepsilon$-balls of the norm $\left\Vert\cdot\right\Vert_\mc{X}$ \cite{talagrand_generic_2005}. Proof of Theorem \ref{thm:SparseLS} follows from the results of \cite{krahmer_suprema_2014}. In particular, we use the following theorem.
 
\begin{theorem}[{Krahmer et al \cite[Theorem 1.4]{krahmer_suprema_2014}}]
\label{thm:RademacherChaos}Let $\mc{A}$ be a symmetric set
of matrices (i.e., $\mc{A}=-\mc{A}$) and $\phi$ be
a vector of i.i.d. Rademacher random variables. Then for some positive
absolute constants $c_{1}$ and $c_{2}$ we have
\begin{align*}
\mbb{P}\left(\!\sup_{\mb{A}\in\mc{A}}\left|\left\Vert \mb{A}\mb{\phi}\right\Vert _{2}^{2}-\mbb{E}\left\Vert \mb{A}\mb{\phi}\right\Vert _{2}^{2}\right|\geq c_{1}E+t\right) &\!\! \leq 2\mr{e}^{-c_{2}\min\left\{\!\frac{t^{2}}{V^{2}},\frac{t}{U}\!\right\}},
\end{align*}
 where $E=\gamma_{2}\left(\mc{A},\left\Vert \cdot\right\Vert \right)\left(\gamma_{2}\left(\mc{A},\left\Vert \cdot\right\Vert \right)+d_{F}\left(\mc{A}\right)\right)+d_{F}\left(\mc{A}\right)d\left(\mc{A}\right)$,
$U=d^{2}\left(\mc{A}\right)$, and $V=d\left(\mc{A}\right)\left(\gamma_{2}\left(\mc{A},\left\Vert \cdot\right\Vert \right)+d_{F}\left(\mc{A}\right)\right)$.
\end{theorem}

\begin{IEEEproof}[Proof of Theorem \ref{thm:SparseLS}.]
It can be easily verified that
\begin{align*}
\widetilde{\mb{A}}_{\mb{x}}\left[\begin{array}{cccc}
\mb{\phi}_{1}^{\mr{T}} & \mb{\phi}_{2}^{\mr{T}} & \dotsm & \mb{\phi}_{K}^{\mr{T}}\end{array}\right]^{\mr{T}}&=\mb{H}\mb{x}.
\end{align*}
 Therefore, the conditioning of $\mb{H}$ on sparse vectors
can be obtained by means of appropriate tail bounds for $\left\Vert \widetilde{\mb{A}}_{\mb{x}}\mb{\phi}\right\Vert _{2}^{2}$
with $\mb{x}\in\mc{D}_{S,L}$ and $\mb{\phi}$
being a vector of $KL$ i.i.d. Rademacher random variables. Our goal
is to apply Theorem \ref{thm:RademacherChaos} for the set $\mc{A}$
defined by \eqref{eq:Set}, to obtain the desired tail bounds. It
follows from Theorem \ref{thm:RademacherChaos} that for some positive
absolute constants $c_{1}$ and $c_{2}$ we have 
\begin{align}
\sup_{\mb{x}\in\mc{D}_{S,L} }\left\Vert \mb{H}\mb{x}\right\Vert _{2}^{2}& \leq\sup_{\mb{x}\in\mc{D}_{S,L} }\mbb{E}\left\Vert \widetilde{\mb{A}}_{\mb{x}}\mb{\phi}\right\Vert _{2}^{2}+c_{1}E+t\nonumber \\
 & =\sup_{\mb{x}\in\mc{D}_{S,L} }K\left\Vert \mb{A}_{\mb{x}}\right\Vert _{F}^{2}+c_{1}E+t\nonumber \\
 & \leq K\theta_{\max}^{2}+c_{1}E+t,\label{eq:SLS-U}
 \end{align}
 and 
 \begin{align}
\inf_{\mb{x}\in\mc{D}_{S,L}}\left\Vert \mb{H}\mb{x}\right\Vert _{2}^{2} & \geq\inf_{\mb{x}\in\mc{D}_{S,L} }\mbb{E}\left\Vert \widetilde{\mb{A}}_{\mb{x}}\mb{\phi}\right\Vert _{2}^{2}-c_{1}E-t\nonumber \\
 & =\inf_{\mb{x}\in\mc{D}_{S,L} }K\left\Vert \mb{A}_{\mb{x}}\right\Vert _{F}^{2}-c_{1}E-t\nonumber \\
 & \geq K\theta_{\min}^{2}-c_{1}E-t\label{eq:SLS-L}
\end{align}
 with probability at least $1-2\mr{e}^{-c_{2}\min\left\{ \frac{t}{U},\frac{t^{2}}{V^{2}}\right\} }$.
It only remains to properly bound the quantities $U$, $V$, and $E$
and choose a reasonable value for $t$.

First we need to bound $\gamma_{2}\left(\mc{A},\left\Vert \cdot\right\Vert \right)$.
Using the special structure of $\widetilde{\mb{A}}_{\mb{x}}$
we deduce that 
\begin{align*}
\left\Vert \widetilde{\mb{A}}_{\mb{x}}-\widetilde{\mb{A}}_{\mb{x}'}\right\Vert  & =\left\Vert \mb{A}_{\mb{x}}-\mb{A}_{\mb{x}'}\right\Vert \\
 & \leq\left\Vert\mb{G}\right\Vert \left\Vert \mb{x}-\mb{x}'\right\Vert _{\infty}\\
 & =\rho\left\Vert \mb{x}-\mb{x}'\right\Vert _{\infty}.
\end{align*}
 Therefore, $N\left(\mc{A},\left\Vert \cdot\right\Vert ,u\right)\leq N\left(\mc{D}_{S,L},\rho\left\Vert \cdot\right\Vert _{\infty},u\right)$.
Then a simple volumetric argument yields 
\begin{align*}
N\left(\mc{D}_{S,L},\rho\left\Vert \cdot\right\Vert _{\infty},u\right) & \leq\binom{L}{S}\left(1+\frac{2\rho}{u}\right)^{S}\\
 & \leq\left(\frac{Le}{S}\left(1+\frac{2\rho}{u}\right)\right)^{S}.
\end{align*}
 Thus, using \eqref{eq:DudleyBound} and for sufficiently large absolute
constant $\beta_{0}$, we can write 
\begin{align*}
\gamma_{2}\left(\mc{A},\left\Vert \cdot\right\Vert \right) & \leq c_{0}\int_{0}^{\theta_{\max}}\sqrt{\log N\left(\mc{A},\left\Vert \cdot\right\Vert ,u\right)}\mr{d}u\\
 & \leq c_{0}\int_{0}^{\theta_{\max}}\sqrt{\log N\left(\mc{D}_{S,L},\rho\left\Vert \cdot\right\Vert _{\infty},u\right)}\mr{d}u\\
 & \leq c_{0}\int_{0}^{\theta_{\max}}\sqrt{S\log\frac{Le}{S}+S\log\left(1+\frac{2\rho}{u}\right)}\mr{d}u\\
 & \leq c_{0}\int_{0}^{\theta_{\max}}\sqrt{S\log\frac{Le}{S}}+\sqrt{S\log\left(1+\frac{2\rho}{u}\right)}\mr{d}u\\
 & \leq c_{0}\theta_{\max}\sqrt{S}\left(\sqrt{\log\frac{Le}{S}}+2\sqrt{\log\left(1+\frac{2\rho}{\theta_{\max}}\right)}\right)\\
 & \leq\theta_{\max}\sqrt{\beta_{0}S\log L},
\end{align*}
 where the last two inequalities follow from the bounds 
 \begin{align*}
 \int_{0}^{\alpha}\sqrt{\log\left(1+\frac{1}{u}\right)}\mr{d}u & \leq 2\alpha \sqrt{\log\left(1+\frac{1}{\alpha}\right)}
 \end{align*}
 with $\alpha=\frac{\theta_{\max}}{2\rho}$ and
 \begin{align*}
 \rho&=\left\Vert \mb{G}\right\Vert \leq\left\Vert \mb{G}\right\Vert _{F}\leq\theta_{\max}\sqrt{L},
 \end{align*}
 respectively. Then we can write 
\begin{align*}
U & =d^{2}\left(\mc{A}\right)=\theta_{\max}^{2}, & V & =d\left(\mc{A}\right)\left(\gamma_{2}\left(\mc{A},\left\Vert \cdot\right\Vert \right)+d_{F}\left(\mc{A}\right)\right)\\
& & & \leq\theta_{\max}^{2}\left(\sqrt{\beta_{0}S\log L}+\sqrt{K}\right),
\end{align*}
and
\begin{align*}
E & =\gamma_{2}\left(\mc{A},\left\Vert \cdot\right\Vert \right)\left(\gamma_{2}\left(\mc{A},\left\Vert \cdot\right\Vert \right)+d_{F}\left(\mc{A}\right)\right)+d_{F}\left(\mc{A}\right)d\left(\mc{A}\right)\\
 & \leq\theta_{\max}^{2}\left(\beta_{0}S\log L+\sqrt{\beta_{0}KS\log L}+\sqrt{K}\right).
\end{align*}
 Setting $t=\frac{\delta\theta_{\min}^{2}}{2}\left(\sqrt{\beta_{0}S\log L}+\sqrt{K}\right)\sqrt{K}$
we have 
\begin{align}
c_{1}E+t & \leq c_{1}\theta_{\max}^{2}\left(\beta_{0}S\log L+\sqrt{\beta_{0}KS\log L}+\sqrt{K}\right)\nonumber\\
&+\frac{\delta\theta_{\min}^{2}}{2}\left(\sqrt{\beta_{0}S\log L}+\sqrt{K}\right)\sqrt{K}\label{eq:Deviation0}
\end{align}
 Recall from the statement of the theorem that $\mu = \frac{\theta_{\max}^{2}}{\theta_{\min}^{2}}$. We would like to upper bound the right-hand side of \eqref{eq:Deviation0} by $\delta K \theta_{\min}^2$. The desired bound can be interpreted as a quadratic polynomial being nonnegative at $\sqrt{K}$. It suffices that $\sqrt{K}$ is greater than the largest root of the polynomial. Therefore, straightforward algebra shows that if $K$ obeys 
\begin{align*}
K & \geq\left(\left(\sqrt{\beta_{0}}+\frac{2c_{1}}{\delta}\cdot\mu\left(\sqrt{\beta_{0}}+2\right)\right)^{2}+\frac{4c_{1}}{\delta}\mu\beta_{0}\right)S\log L\\
&=O\left(\delta^{-2}\mu^{2}S\log L\right),
\end{align*}
 then the right-hand side of \eqref{eq:Deviation0} is  bounded from above by $\delta K\theta_{\min}^2$ and thus
\begin{align}
c_{1}E+t & \leq\delta K\theta_{\min}^{2}.\label{eq:Deviation}
\end{align}
 Furthermore, we have 
\begin{align*}
\frac{t}{U} & \geq\frac{\delta}{2\mu}K
\end{align*}
 and 
\begin{align*}
\frac{t}{V} & \geq\frac{\delta}{2\mu}\sqrt{K}
\end{align*}
which imply 
\begin{align}
\mr{e}^{-c_{2}\min\left\{ \frac{t^{2}}{V^{2}},\frac{t}{U}\right\} } & \leq \mr{e}^{-c_{2}K\min\left\{ \frac{\delta}{2\mu},\frac{\delta^{2}}{4\mu^2}\right\}}.\label{eq:FailureProb}
\end{align}
Applying \eqref{eq:Deviation} to \eqref{eq:SLS-U} and \eqref{eq:SLS-L}
and bounding the tail probability using \eqref{eq:FailureProb} shows
that  with probability at least $1-2\mr{e}^{-c_{2}K\min\left\{ \delta/\left(2\mu\right),\delta^{2}/\left(2\mu\right)^2\right\}}$ the inequalities 
\begin{alignat*}{2}
\left(1-\delta\right)K\theta_{\min}^{2} & \leq\left\Vert \mb{H}\mb{x}\right\Vert _{2}^{2} & \leq K\left(\theta_{\max}^{2}+\delta\theta_{\min}^{2}\right),
\end{alignat*}
 which are equivalent to \eqref{eq:SLS-bound}, hold for all $\mb{x}\in\mc{D}_{S,L}$.
 \end{IEEEproof}

\bibliographystyle{IEEEtran}
\bibliography{references}

\end{document}